\documentclass[10pt,a4paper]{article}

\usepackage[latin3]{inputenc}
\usepackage{amsmath}
\usepackage{amsfonts}
\usepackage{amssymb}
\usepackage{amsthm}

\usepackage{mathtools}
\usepackage{bussproofs}
\usepackage{cmll}
\usepackage{stmaryrd}

\usepackage{tikz}
\usetikzlibrary{arrows}
\usetikzlibrary{positioning}

\newcommand{\bind}{>\!\!>\!\!=}
\newcommand{\dial}[3]{\left| #1 \right|^{#2}_{#3}}

\newcommand{\LL}{\operatorname{\mathbf{LL}}}
\newcommand{\MLL}{\operatorname{\mathbf{MLL}}}

\newcommand{\MELL}{\operatorname{\mathbf{MELL}}}
\newcommand{\Mod}{\operatorname{\mathbf{-Mod}}}

\newcommand{\aeq}{\rotatebox[origin = c]{180}{{\sf \AE}}}
\newcommand{\sembrack}[1]{\llbracket #1 \rrbracket}

\newcommand{\id}{\operatorname{id}}
\newcommand{\cg}{\mathcal G}
\newcommand{\ch}{\mathcal H}
\newcommand{\R}{\mathcal R}
\newcommand{\fd}{\mathfrak D}
\newcommand{\HO}{\mathbf{HO}}

\newcommand{\ignore}[1]{}

\newtheorem{lemma}{Lemma}
\newtheorem{theorem}{Theorem}
\newtheorem{definition}{Definition}
\newtheorem{corollary}{Corollary}

\title{Dialectica models of additive-free linear logic}
\author{Jules Hedges \\
Queen Mary University of London \\
\texttt{j.hedges@qmul.ac.uk}}
\date{}

\begin{document}

\maketitle

\begin{abstract}
This paper presents a construction which transforms categorical models of additive-free propositional linear logic, closely based on de Paiva's dialectica categories and Oliva's functional interpretations of classical linear logic. The construction is defined using dependent type theory, which proves to be a useful tool for reasoning about dialectica categories. Abstractly, we have a closure operator on the class of models: it preserves soundness and completeness and has a monad-like structure. When applied to categories of games we obtain `games with bidding', which are hybrids of dialectica and game models, and we prove completeness theorems for two specific such models.
\end{abstract}

\section{Introduction}

This paper presents a construction which transforms categorical models of additive-free propositional linear logic, closely based on de Paiva's dialectica categories and Oliva's functional interpretations of classical linear logic.

The dialectica categories \cite{depaiva91a} are a family of models of intuitionistic logic, and classical and intuitionistic linear logic, based on G\"odel's dialectica interpretation. Historically they were the first models of linear logic to not equate multiplicative and additive units, and they have been generalised in several ways, for example \cite{hyland02} defines dialectica categories starting only from a partially ordered fibration. The construction in this paper is closely related to \cite{depaiva91b} and \cite{depaiva07}; the similarities and differences between that construction and the original dialectica categories is discussed in those papers. While most of the literature on dialectica categories aims to construct large classes of structured categories and then characterise those which are sound models of some logic, the aim of this paper is rather different: to construct a small number of concrete models which can be interpreted as game models and are amenable to a proof-theoretic analysis of the valid formulas, and in particular are as close as possible to being complete models of linear logic.

Based on de Paiva's models, \cite{shirahata06} gave a syntactic dialectica and Diller-Nahm interpretation to first order affine logic, and \cite{oliva07b} to classical linear logic. The semantics of the Diller-Nahm variant is explored in detail in chapter 4 of \cite{depaiva91a}, and will be used in this paper. A completeness theorem is given in \cite{oliva08} for the dialectica interpretation, based on G\"odel's original completeness theorem for Heyting arithmetic \cite{avigad98}, which has not been exploited so far in the semantic literature. This relies on a small but crucial modification to de Paiva's interpretation of the linear exponentials. The Diller-Nahm interpretation of linear logic appears in \cite{oliva07b} and \cite{oliva10}, although no completeness proof for the Diller-Nahm interpretation of linear logic appears in the literature, to the author's knowledge.

The dialectica interpretation, intuitively, is a proof translation which takes a formula $\varphi$ to a quantifier-free formula $\dial{\varphi}{x}{y}$ in which the variables $x$ and $y$ appear free. The variable $x$ represents `witnesses', or evidence that a theorem is true, and $y$ represents `counter-witnesses', or evidence that a theorem is false. The validity of a theorem is then reduced to the existence of a witness which defeats every counter-witness, that is, $\exists x \forall y . \dial{\varphi}{x}{y}$. However even if $\varphi$ is a first-order formula the variables $x$ and $y$ may have higher types. The original purpose was to prove the relative consistency of Heyting arithmetic to the quantifier-free language called system T, however the dialectica interpretation is now mainly used to give a computational interpretation to theorems of classical analysis, see \cite{kohlenbach08}. 

The semantic equivalent to the dialectica interpretation, at least from the point of view of this paper, is to replace the formula $\dial{\varphi}{x}{y}$ with a double-indexed family of objects in some model $\R$. We can imagine that we are composing the syntactic proof translation with a semantic interpretation of formulas. The fact that the dialectica interpretations of linear negation and multiplicative conjunction are given recursively by
\[ \dial{\varphi^\bot}{y}{x} = \left( \dial{\varphi}{x}{y} \right)^\bot \]
\[ \dial{\varphi \otimes \psi}{x,u}{f,g} = \dial{\varphi}{x}{fu} \otimes \dial{\psi}{u}{gx} \]
(in particular, that the same connectives occur on the right hand side) tells us that $\R$ must have a sound interpretation of these connectives. This leads us to the construction in \cite{depaiva91b}, which builds a dialectica category from a posetal model of multiplicative linear logic, or \emph{lineale} \cite{depaiva02}. The dialectica interpretation eliminates additives (in the sense that additives do not appear on the right hand side of the corresponding formulas), and it is also possible to eliminate exponentials in a sound way by defining
\[ \dial{\oc \varphi}{x}{f} = \dial{\varphi}{x}{fx} \]
This is the interpretation of exponentials used in all of the literature on dialectica categories, and also in \cite{shirahata06}. However the completeness theorem of \cite{oliva08} relies on changing this definition to
\[ \dial{\oc \varphi}{x}{f} = \oc \dial{\varphi}{x}{fx} \]
To interpret this semantically $\R$ must also have a sound interpretation of the exponential, which leads to our construction of dialectica categories beginning from an arbitrary model of multiplicative-exponential linear logic ($\MELL$). Thus this work can be seen as the result of a `dialogue' between syntax and semantics.

Overall, we have a construction $\fd$ which takes a model of $\MELL$ to a model of $\LL$. The first of two aims of this paper is to explore the abstract properties of $\fd$. We prove in section \ref{functor} that $\fd$ is functorial, and in section \ref{monad} it has a monad-like structure on a particular category of models of $\MLL$, although one of the monad laws fails and even the weaker result fails to extend to $\MELL$. This is closely related to the main theorem in \cite{hofstra11}. (We could also explore the 2-categorical properties of $\fd$, but that is left for later work.) We also prove that $\fd$ preserves soundness (section \ref{sound}) and completeness (section \ref{complete}) for $\MELL$, so we can justify calling it a `closure operator' on models. 

The second aim of this paper is to construct specific dialectica categories which have logical completeness properties. This requires that the underlying model also has completeness properties, which in practice means constructing a dialectica category from a category of games. In section \ref{games} we informally describe such a dialectica category as a category of `games with bidding', greatly extending the comments in \cite{blass91} on viewing dialectica categories as game models. In particular in section \ref{complete} we consider `Hyland-Ong games with bidding' based on \cite{hyland93}, and `asynchronous games with bidding' based on \cite{mellies04}, and prove that these models are complete respectively for $\MLL$ and $\MELL$.

The model of asynchronous games with bidding, in particular, is an extremely interesting model because the starting model has the strongest possible completeness theorem, namely it is fully complete for $\MELL$. An analysis of the formulas containing additives which are valid in this model will be carried out in a follow-up paper, but an overview of the argument is given in section \ref{additives}. Also in this section we give a simple counterexample proving that there is no dialectica category which is both sound and complete for full propositional $\LL$. 

There are two main technical ideas in this paper which contribute to our two aims. The first is that we replace the posets of \cite{depaiva91b} and \cite{hyland01} with categories, and use dependent type theory in defining and reasoning about our models. If our metatheory has choice this formally gains nothing, however in practice dependent type theory proves to be a powerful tool. This will be justified in particular in sections \ref{complete} and \ref{monad}, which would be hard to formalise without dependent type theory. It also suggests the implementation of this construction (and the formalisation of the proofs in this paper) in a dependently typed programming language. This would require libraries for 2-category theory and monoidal category theory, and would be an interesting way to embed linear reasoning into a proof assistant.

The second idea is that we work with the linear-nonlinear semantics of $\MELL$ and $\LL$ given in \cite{benton94}. This allows the relationship between the linear and intuitionistic dialectica categories to be clearly seen, and allows us to factor the exponential into four parts. This also suggests turning back around to syntax and studying a syntactic dialectica interpretation of linear-nonlinear logic.

Note that in this paper we are only considering \emph{classical} linear logic. The differences between dialectica models of classical and intuitionistic linear logic are subtle: firstly for intuitionistic linear logic the sets of witnesses and counterexamples must \emph{both} be nonempty, whereas for classical linear logic one may be empty; and secondly for intuitionistic linear logic we consider the bids in games with bidding to be sequential rather than simultaneous. Since the two logics coincide in the absence of additives, the difference will not often affect us.

\section{The dialectica transformations of a category}

In this section we will define the two dialectica transformations of a category, and relate them to the existing literature on dialectica categories. The game-semantic intuition corresponding to these definitions will be given in section \ref{games}.

Let $R$ be an arbitrary category. We will define a category $\fd_l (R)$ called the \emph{linear dialectica transformation} of $R$. The objects of $\fd_l (R)$ are double-indexed families $\cg^X_Y$ where $X$ and $Y$ are arbitrary sets not both empty, and each $\cg^x_y$ is an object of $R$. Throughout this paper we will specify such objects using the notation
\[ \cg^X_Y : \binom{x}{y} \mapsto \cdots \]
where the right hand side is an expression in terms of $x$ and $y$. Since $X$ and $Y$ will often be (dependent) pairs we will drop the parentheses, as is done in the proof theory literature. Sometimes we will decorate witness and counter-witness variables with their individual types for clarity, as in
\[ \cg^{X \times U}_{Y \times V} : \binom{x : X, u : U}{y : Y, v : V} \mapsto \cdots \]

A morphism from $\cg^X_Y$ to $\ch^U_V$ is an element of a dependent type in the category of sets:
\[ \hom_{\fd_l (R)} \left( \cg^X_Y, \ch^U_V \right) = \sum_{\substack{f : X \to U \\ g : V \to Y}} \prod_{\substack{x : X \\ v : V}} \hom_R \left( \cg^x_{gv}, \ch^{fx}_v \right) \]
Hence a morphism is a triple $(f, g, \alpha)$ where $f : X \to U$, $g : V \to Y$ and $\alpha$ is a double-indexed family of $R$-morphisms
\[ \alpha_{x,v} : \cg^x_{gv} \to \ch^{fx}_v \]
The proof-theoretic reading of this is that a morphism consists of a witness, together with a mapping that takes each counter-witness to a proof that the counter-witness is invalid. This is simply the type-theoretic interpretation of the usual dialectica interpretation of linear implication, with quantifiers replaced by dependent types.

For simplicity, in this paper we only explicitly use the set-theoretic interpretation of dependent type theory, however it should be straightforward to generalise to any model of dependent type theory. This would require $R$ to be enriched over a locally cartesian closed category $\mathcal C$, and that we have a suitable fibration of objects of $R$ over $\mathcal C$ to replace set-indexed families, similar to \cite{hyland02} (this idea was suggested in \cite{hyland07}). 

In $\fd_l (R)$ the identity morphism on $\cg^X_Y$ is given by the identity functions on $X$ and $Y$ together with identity morphisms in $R$. The composition of a morphism $\cg^X_Y \multimap \ch^U_V$ given by $(f,g,\alpha)$ and another $\ch^U_V \multimap \mathcal I^P_Q$ given by $(f',g',\beta)$ is given by $f' \circ f : X \to P$ and $g \circ g' : Q \to Y$, together with the composition
\[ (\beta \circ \alpha)_{x,q} = \beta_{fx,v} \circ \alpha_{x,g'q} : \hom_R \left( \cg^x_{g (g'q)}, \mathcal I^{f' (fx)}_q \right) \]

\begin{lemma}
Let $R$ be any category, then $\fd_l (R)$ is a category with finite products and coproducts.
\end{lemma}

\begin{proof}
By proposition 3.7 of \cite{hyland01}.
\end{proof}

Using the axiom of choice (at least in the case $\mathcal C = \mathbf{Set}$), this definition is equivalent to $M_N (\mathcal C)$ in \cite{depaiva91b} where $N$ is the posetal reflection of $R$ (assuming a Grothendeick universe, since $R$ will be large in general). To be clear, this definition is not intended to be exactly equivalent to the original dialectica categories in \cite{depaiva91a}, which is more elegant and far more general but is hard to use for concrete calculations. In particular using type theory gives us explicit names for all of our morphisms, and this will make our life easier especially in sections \ref{complete} and \ref{monad}. Moreover we can avoid using the axiom of choice in our metatheory, and so the contents of this paper could be directly implemented in a dependently typed programming language.

Next we will construct the Diller-Nahm translation $\fd_i (S)$ of an arbitrary category $S$ with finite products. This construction is mostly closely related to that in \cite{hyland02}, although we consider it in far less generality than in that paper. The objects of $\fd_i (S)$, as before, are double-indexed families $\cg^X_Y$ where $X$ and $Y$ are sets not both empty and each $\cg^x_y$ is an element of $S$. The hom-sets are defined by
\[ \hom_{\fd_i (S)} \left( \cg^X_Y, \ch^U_V \right) = \sum_{\substack{f : X \to U \\ g : X \times V \to Y^*}} \prod_{\substack{x : X \\ v : V}} \hom_S \left( \prod_{y \in g(x,v)} \cg^x_y, \ch^{fx}_v \right) \]
Here $Y^*$ is the set of finite multisets with elements in $Y$. This definition is the type-theoretic interpretation of the Diller-Nahm interpretation of intuitionistic implication
\[ \exists f^{X \to U}, g^{X \times V \to Y^*} \forall x^X, v^V . \left( \forall y \in g(x,v) . \dial{\varphi}{x}{y} \right) \to \dial{\psi}{fx}{v} \]
However we carefully distinguish `internal' and `external' quantifiers: the internal $\forall$ is interpreted as the categorical product in the underlying model, and the external $\exists \forall$ is interpreted as dependent types in $\mathcal C$.

In $\fd_i (S)$ the structure is very similar. If we have a morphism given by $f : X \to U$ and $g : X \times V \to Y^*$ and another given by $f' : U \to P$ and $g' : U \times Q \to V^*$ the composition is given by $f' \circ f : X \to P$ and
\[ \lambda x^X, q^V . g' (fx, q) \bind \lambda v^V . g(x,v) : X \times Q \to Y^* \]
together with composition in $S$. Here $\bind$ is the bind operator of the finite multiset  monad, where $l \bind f$ applies $f$ to each element of $l$, each giving a multiset, and collects the results with a union.

\begin{lemma}
Let $S$ be any category with finite products, then $\fd_i (S)$ is a category with finite products.
\end{lemma}

\begin{proof}
By section 3 of \cite{hyland02}.
\end{proof}

\section{The dialectica transformation of a linear-nonlinear adjunction}\label{sound}

We begin with a general definition of a model of $\MELL$ and a model of $\LL$. A model of multiplicative linear logic ($\MLL$) is given by a $*$-autonomous category $R$ \cite{barr91}, that is, a symmetric monoidal closed category $(R, \otimes, \multimap, 1)$ with a functor $^\bot : R \to R$ and natural isomorphisms $^\bot \circ ^\bot \cong \operatorname{id}_R$ and
\[ \hom_R (X \otimes Y, Z^\bot) \cong \hom_R (X, (Y \otimes Z)^\bot) \]

For the interpretation of exponentials we use the linear-nonlinear semantics of \cite{benton94}, which is surveyed in detail in \cite{mellies09}. A categorical model of $\MELL$ is given by a $*$-autonomous category $R$ together with another category $S$ with finite products and an adjunction
\begin{center}
\begin{tikzpicture}[node distance=2cm]
\node (A) {$S$};
\node (B) [right of=A] {$\bot$};
\node (C) [right of=B] {$R$};
\draw [->] (A) edge [bend left=45] node [below] {$L$} (C);
\draw [->] (C) edge [bend left=45] node [above] {$M$} (A);
\end{tikzpicture}
\end{center}
or, more briefly,
\[ L \dashv M : R \to S \]
Here $L$ (called \emph{linearisation}) and $M$ (called \emph{multiplication}) are lax symmetric monoidal functors, that is, there are natural transformations
\begin{align*}
M(X) \times M(Y) &\to M(X \otimes Y) &\top &\to M(1) \\
L(X) \otimes L(Y) &\to L(X \times Y) &1 &\to L (\top)
\end{align*}
and the unit and counit of the adjunction must also respect the monoidal and cartesian monoidal structures (ie. the adjunction must be a \emph{symmetric monoidal adjunction}). Such a setup is called a \emph{linear-nonlinear adjunction}. Given this adjunction, the denotation of the exponential $!$ is the composition $L \circ M$, which is a comonad on $R$ (and conversely, if we have a model in which $\oc$ is given explicitly we can recover $S$, $M$ and $L$ from the co-Kleisli adjunction). The entire model, which contains a pair of categories and functors and various natural transformations, will be denoted $\R$. For a model of $\LL$ we simply require that $R$ also has finite products.

Given such a model of $\MELL$,  the dialectica transformation of this model will be a new pair of categories and a linear-nonlinear adjunction
\ignore{\[ \fd_m (L) \dashv \fd_f (M) : \fd_l (R) \to \fd_i (S) \]}
\begin{center}
\begin{tikzpicture}[node distance=2cm]
\node (A) {$\fd_i (S)$};
\node (B) [right of=A] {$\bot$};
\node (C) [right of=B] {$\fd_l (R)$};
\draw [->] (A) edge [bend left=45] node [below] {$\fd_{dn} (L)$} (C);
\draw [->] (C) edge [bend left=45] node [above] {$\fd_f (M)$} (A);
\end{tikzpicture}
\end{center}
The categories $\fd_l (R)$ and $\fd_i (S)$ are precisely the categories defined in the previous section. The transformations of the functors $M$ and $L$ will be given below. The transformed model as a whole will be denoted $\fd (\mathcal R)$.

The interpretations of each connective in $\fd_l (R)$ is given in figure \ref{connectives}.
\begin{figure}
\center{\textbf{Multiplicatives}}
\begin{align*}
1^{\{ * \}}_{\{ * \}} &: \binom{*}{*} \mapsto 1 \\
\\
\bot^{\{ * \}}_{\{ * \}} &: \binom{*}{*} \mapsto \bot \\
\\
(\cg^X_Y)^\bot = (\cg^\bot)^Y_X &: \binom{y}{x} \mapsto (\cg^x_y)^\bot \\
\\
\cg^X_Y \otimes \ch^U_V = (\cg \otimes \ch)^{X \times U}_{(U \to Y) \times (X \to V)} &: \binom{x,u}{f,g} \mapsto \cg^x_{fu} \otimes \ch^u_{gx} \\
\\
\cg^X_Y \parr \ch^U_V = (\cg \parr \ch)^{(V \to X) \times (Y \to U)}_{Y \times V} &: \binom{f,g}{y,v} \mapsto \cg^{fv}_y \parr \ch^{gy}_v \\
\end{align*}

\center{\textbf{Additives}}
\begin{align*}
\top^{\{ * \}}_\varnothing & && \\
\\ 
0^\varnothing_{\{ * \}} \\
\\
\cg^X_Y \with \ch^U_V = (\cg \with \ch)^{X \times U}_{Y + V} &: \binom{x,u}{z} \mapsto \begin{dcases*}
\cg^x_z & if $z \in Y$ \\
\ch^u_z & if $z \in V$ \\
\end{dcases*} \\
\\
\cg^X_Y \oplus \ch^U_V = (\cg \oplus \ch)^{X + U}_{Y \times V} &: \binom{z}{y,v} \mapsto \begin{dcases*}
\cg^z_y & if $z \in X$ \\
\ch^z_v & if $z \in U$ \\
\end{dcases*} \\
\end{align*}

\center{\textbf{Exponentials}}
\begin{align*}
\oc \cg^X_Y = (\oc \cg)^X_{X \to Y^*} &: \binom{x}{f} \mapsto \bigotimes_{y \in fx} \oc \cg^x_y \\
\\
\wn \cg^X_Y = (\wn \cg)^{Y \to X^*}_Y &: \binom{g}{y} \mapsto \bigparr_{x \in gy} \wn \cg^x_y \\
\end{align*}
\caption{Interpretation of constants and connectives in $\fd_l (R)$}
\label{connectives}
\end{figure}

\begin{lemma}
Let $R$ be any $*$-autonomous category, then $\fd_l (R)$ is a $*$-autonomous category.
\end{lemma}

\begin{proof}
By propositions 3.6 of \cite{hyland01}.
\end{proof}

Now we give the dialectica transformations $\fd_f (M)$ and $\fd_{dn} (L)$ of the multiplication and linearisation functors. The operation $\fd_f$ is a straightforward lifting operation. The subscript $f$ stands for \emph{functor} since this construction will be used in section \ref{functor} to give the action of $\fd$ on maps (or functors) of models. Suppose the multiplication functor is $M : R \to S$. The functor $\fd_f (M) : \fd_l (R) \to \fd_i (S)$ acts on objects $\cg^X_Y$ of $\fd_l (R)$ by
\[ (\fd_f (M) (\cg))^X_Y : \binom{x}{y} \mapsto M (\cg^x_y) \]
For the action of $\fd_f (M)$ on morphisms, suppose we have a morphism of $\fd_l (R)$ from $\cg^X_Y$ to $\ch^U_V$ given by $(f,g,\alpha)$ where $f : X \to U$, $g : V \to Y$ and $\alpha_{x,v} : \hom_R \left( \cg^x_{gv}, \ch^{fx}_v \right)$. We need to find an element of
\[ \sum_{\substack{f' : X \to U \\ g' : X \times V \to Y^*}} \prod_{\substack{x : X \\ v : V}} \hom_S \left( \prod_{y \in g'(x,v)} M (\cg^x_y), M (\ch^{f'x}_v) \right) \]
We take $f' = f$ and $g'(x,v)$ to be the multiset containing only $gv$. Then
\[ \prod_{y \in g'(x,v)} M (\cg^x_y) = M (\cg^x_{gv}) \]
and so $M (\alpha_{x,v})$ is a morphism of the correct type.

Suppose the linearisation functor is $L : S \to R$. The functor $\fd_{dn} (L) : \fd_i (S) \to \fd_l (R)$ acts on objects $\cg^X_Y$ by
\[ (\fd_{dn} (L) (\cg))^X_{X \to Y^*} : \binom{x}{f} \mapsto \bigotimes_{y \in fx} L (\cg^x_y) \]
Here $\bigotimes_{y \in fx}$ is the fold of the monoidal product of $R$ over the finite multiset $fx$, where the fold over the empty multiset is the unit $1 \in R$. The subscript $dn$ stands for \emph{Diller-Nahm}, since this definition contains the essence of the Diller-Nahm functional interpretation. The intuitive justification for this definition is that the exponential $\fd_{dn} (L) \circ \fd_f (M)$ should be an interpretation of
\[ \oc \forall y \in fx . \dial{\varphi}{x}{y} \]
which is the Diller-Nahm interpretation of the exponentials in \cite{oliva07b}. Since we are working over set theory we `know' the (finite) size of $fx$, so we can replace the $\forall$ with a folded $\with$. (This is a subtle point: we are simply defining a family of formulas, whereas when using free variables a formula must have a fixed structure.) Then we use the fact that $\oc$ is strong monoidal (the `transmutation principle' of linear logic, see section 7.1 of \cite{mellies09}) to obtain
\[ \bigotimes_{y \in fx} \oc \dial{\varphi}{x}{y} \]
When this is factored as
\[ \bigotimes_{y \in fx} L \left( M \dial{\varphi}{x}{y} \right) \]
the $M$ becomes absorbed into the definition of $\fd_f (M)$, and we are left with $\fd_{dn} (L)$. (We could write it instead as $L \with$, but using $\otimes L$ gives the exponential in figure \ref{connectives} directly. Taking the exponential to be $\otimes \oc$ is preferable to $\oc \with$ because we need not assume that $L$ has products.)

Now suppose we have a morphism of $\fd_i (S)$ from $\cg^X_Y$ to $\ch^U_V$ given by $(f,g,\alpha)$ where $f : X \to U$, $g : X \times V \to Y^*$ and
\[ \alpha_{x,v} : \hom_S \left( \prod_{y \in g(x,v)} \cg^x_y, \ch^{fx}_v \right) \]
We need to find an element of
\[ \hom_{\fd_l (R)} \left( (\fd_{dn} (L) (\cg))^X_{X \to Y^*}, (\fd_{dn} (L) (\ch))^U_{U \to V^*} \right) \]
The witnesses are $f : X \to U$ and $g' : (U \to V^*) \to (X \to Y^*)$ given by
\[ g' = \lambda h^{U \to V^*}, x^X . h(fx) \bind \lambda v^V . g(x,v) \]
Given $x \in X$ and $h : U \to V^*$ we need to find an element of
\[ \hom_R \left((\fd_{dn} (L) (\mathcal G))^x_{g'h}, (\fd_{dn} (L) (\mathcal H))^{fx}_h \right) = \hom_R \left( \bigotimes_{y \in g'hx} L (\mathcal G^x_y), \bigotimes_{v \in h(fx)} L (\mathcal H^{fx}_v) \right) \]
We have
\[ \bigotimes_{v \in h(fx)} L (\alpha_{x,v}) : \hom_R \left( \bigotimes_{v \in h(fx)} L \left( \prod_{y \in g(x,v)} \mathcal G^x_y \right), \bigotimes_{v \in h(fx)} L (\mathcal H^{fx}_v) \right) \]
Here we can use that $L$ is a symmetric monoidal functor to get an element of
\[ \hom_R \left( \bigotimes_{v \in h(fx)} \bigotimes_{y \in g(x,v)} L (\mathcal G^x_y), \bigotimes_{v \in h(fx)} L (\mathcal H^{fx}_v) \right) \]
Finally the left hand side can be written as a single monoidal product over $y \in g'hx$ by definition of the monadic bind.

\begin{lemma}
$\fd_{dn} (L) \dashv \fd_f (M) : \fd_l (R) \to \fd_i (S)$ is a linear-nonlinear adjunction.
\end{lemma}

\begin{proof}
By proposition 14 of \cite{mellies09} it suffices to prove that $\fd_{dn} (L) \dashv \fd_f (M)$ is an adjunction and $\fd_{dn} (L)$ is strong symmetric monoidal.

The equation for the adjunction is
\[ \hom_{\mathfrak D_l (R)} \left( \mathfrak D_{dn} (L) (\mathcal G^X_Y), \mathcal H^U_V \right) \cong \hom_{\mathfrak D_i (S)} \left( \mathcal G^X_Y, \mathfrak D_f (M) (\mathcal H^U_V) \right) \]
We evaluate
\[ \hom_{\fd_l (R)} \left (\fd_{dn} (L) (\cg^X_Y), \ch^U_V \right) = \sum_{\substack{f : X \to U \\ g : V \to (X \to Y^*)}} \prod_{\substack{x : X \\ v : V}} \hom_R \left( \bigotimes_{y \in gvx} L (\cg^x_y), \ch^{fx}_v \right) \]
and
\[ \hom_{\fd_i (S)} \left( \cg^X_Y, \fd_f (M) (\ch^U_V) \right) = \sum_{\substack{f : X \to U \\ g : X \times V \to Y^*}} \prod_{\substack{x : X \\ v : V}} \hom_S \left( \prod_{y \in g(x,v)} \cg^x_y, M (\ch^{fx}_v) \right) \]
These are isomorphic using $L \dashv M$ and the fact that $L$ is strong monoidal.

To prove that $\fd_{dn} (L)$ is strong monoidal we must show that
\[ \fd_{dn} (L) (\cg^X_Y) \otimes \fd_{dn} (L) (\ch^U_V) \cong \fd_{dn} (L) (\cg^X_Y \with \ch^U_V) \]
We evaluate
\[ \left( \fd_{dn} (L) (\cg) \otimes \fd_{dn} (L) (\ch) \right)^{X \times U}_{(X \times U \to Y^*) \times (X \times U \to V^*)} : \binom{x,u}{f,g} \mapsto \bigotimes_{y \in f(x,u)} L (\cg^x_y) \otimes \bigotimes_{v \in g(x,u)} L (\ch^u_v) \]
and
\[ \fd_{dn} (L) (\cg \with \ch)^{X \times U}_{X \times U \to (Y + V)^*} : \binom{x,u}{h} \mapsto \bigotimes_{z \in h(x,u)} \begin{dcases*}
L (\cg^x_z) & if $z \in Y$ \\
L (\ch^u_z) & if $z \in V$
\end{dcases*} \]
These are isomorphic due to the natural isomorphism $Y^* \times V^* \cong (Y + V)^*$ (note that this isomorphism does not hold if we replace finite multisets with finite ordered lists, ie. free commutative monoids by free noncommutative monoids). Finally, the symmetry of $\fd_{dn} (L)$ also inherits easily from that of $L$.
\end{proof}

We can therefore derive the interpretation of $!$ as the composition $\fd_{dn} (L) \circ \fd_f (M)$. Given an object $\cg^x_y$, its exponential is
\[ (! \cg)^X_{X \to Y^*} : \binom{x}{f} \mapsto \bigotimes_{y \in fx} ! \cg^x_y \]
where the exponential in the underlying model is $! = L \circ M$.

It is worth noting that, as in chapter 4 of \cite{depaiva91a}, the functor $\fd_{dn} (L)$ factors into three parts $\fd_{dn} (L) = B \circ A \circ \fd_f (L)$ where $A$ and $B$ (called $T$ and $S$ in \cite{depaiva91a}) are endofunctors on $\fd_l (R)$ given respectively by
\[ (A (\mathcal G))^X_{Y^*} : \binom{x}{s} \mapsto \bigotimes_{y \in s} \mathcal G^x_y \]
and
\[ (B (\mathcal G))^X_{X \to Y} : \binom{x}{f} \mapsto \mathcal G^x_{fx} \]
We can interpret $A$ and $B$ game-semantically as giving two different advantages to Abelard. $A$ allows Abelard to play several moves, and $B$ allows Abelard to observe Eloise's move. Both of these are expressed by monads on the category of sets, respectively the finite multiset monad and the reader monad $(X \to)$. The exponential of $\fd_l (R)$ therefore factors into four parts as
\[ B \circ A \circ \fd_f (L) \circ \fd_f (M) \]
The functors $A$ and $B$ have much structure in their own right: they are both comonads on $\fd_l (R)$ with a distributivity law between them making $B \circ A$ into another comonad. However $B \circ A$ is a linear exponential comonad (which is a direct categorical semantics of the exponential, see \cite{hyland01}), whereas $A$ and $B$ individually are not. The entire reason we also compose with $\fd_f (L) \circ \fd_f (M) = \fd_f (L \circ M)$, which after all requires more structure in the underlying model, is to obtain the completeness theorem in section \ref{complete}.

The lemmas in this section add up to a soundness theorem.

\begin{theorem}
If $\mathcal R$ is a sound model of $\MELL$ then $\fd (\mathcal R)$ is a sound model of $\LL$.
\end{theorem}

\section{Games with bidding}\label{games}

In section \ref{complete} we will investigate applying the transformation $\fd$ to models which are complete (that is truth implies provability, which is a weaker property than \emph{full completeness} which is more often considered). In practice this means letting $\mathcal R$ be a game model. In this section we give some general remarks about $\fd (\mathcal R)$ when $\mathcal R$ is a game model.

In general, a game model is a category $R$ whose objects are games, and whose morphisms are (relative) winning strategies. Thus logically formulas are denoted by games and proofs by winning strategies. The denotation of linear negation is interchange of players (at least for classical linear logic), and the denotation of $\otimes$ is some form of concurrent play, making $R$ into a *-autonomous category. For models which have additives the product $\mathcal G \with \mathcal H$ is usually denoted by a game in which Abelard chooses which of the two games will be played, and for $\mathcal G \oplus \mathcal H$ Eloise makes the choice. The exponential is often similar to an infinite tensor product. The point of making these informal observations is that they are preserved under the transformation $\fd$.

We begin by considering the two-element boolean algebra $\mathbb B$ as a degenerate game model containing only two games: one which Eloise wins immediately, and one which Abelard wins immediately. Thus we can see $\fd (\mathbb B)$, which is called $\mathbf G (\mathcal C)$ in the terminology of \cite{depaiva91a} (where $\mathcal C$ is the category of sets or another suitable model of dependent type theory), as a model of games with bidding in which the games contain only the bidding round, and after the bidding round one player is declared to have won. The possibility of viewing dialectica categories as categories of games has been discussed in several places, and in particular in the final section of \cite{blass91}, and this section greatly extends that idea.

One issue with viewing dialectica categories as games is the strange `causality' in a game such as $\cg^X_Y \multimap \ch^U_V$, in which $u$ depends on $x$ but not $v$, and $y$ depends on $v$ but not $x$. One way to view the strange dynamics of this game is as a generalisation of history-freeness in which the moves are chosen in the order $(x,u,v,y)$, where Abelard's strategy to choose $x,v$ is history-free and Eloise's strategy may depend on the most recent move but not the remainder of the history. Alternatively we can imagine the bidding round to be played by two teams of two players (like Bridge) with a particular message-passing protocol
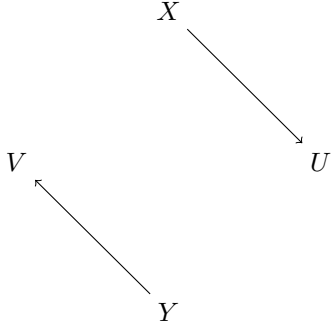
\begin{center}\begin{tikzpicture}[node distance=2cm, auto]
\node (X) {$X$};
\node (C) [below of=X] {};
\node (U) [right of=C] {$U$};
\node (Y) [below of=C] {$Y$};
\node (V) [left of=C] {$V$};
\draw [->] (X) to node {} (U);
\draw [->] (Y) to node {} (V);
\end{tikzpicture}\end{center}
Partners sit opposite each other, with $X$ and $Y$ representing Abelard and $U$ and $V$ representing Eloise, and the arrows representing the direction of message-passing. Unfortunately both of these intuitions (history-freeness and message-passing) break down when we consider higher order bids (that is, bids which are functions depending on other functions). There is a general but less satisfactory intuition in these cases: the players submit (higher order) computer programs, which are finite representations of their strategy, to play on their behalf.

Now we consider informally a `general non-degenerate game model'. The conclusion is that the construction $\fd$, which can be applied to any model, preserves the property of `being a game model'. For a concrete game model these informal remarks could be made precise: the simplest example is the category of Blass games of \cite{blass91}; in section \ref{complete} we consider the category of Hyland-Ong games \cite{hyland93} and the category of asynchronous games of \cite{mellies04}.

An object $\mathcal G^X_Y$ of $\fd_l (R)$ consists of sets of bids $X$ and $Y$ for Eloise and Abelard, together with a game $\mathcal G^x_y$ in the underlying model for each pair of bids. Thus a winning strategy for Eloise consists of a bid $x \in X$, together with a winning strategy $\sigma_y$ for $\mathcal G^x_y$ for every bid $y$ of Abelard. Thus $\mathcal G^X_Y$ can be seen as a game with bidding: first Eloise and Abelard simultaneously bid, and then the pair of chosen bids determines precisely which subsequent game will be played. (Very informally this is somewhat like the game of Bridge: there is an initial bidding round which determines exactly which variant of Whist will be played.)

The negation of $\fd_l (R)$ is to interchange players in the bidding round and then apply the negation of $R$. Thus when $R$ is a game model the negation of $\fd_l (R)$ overall is simply interchange of players in the compound game. The other connectives which behave very cleanly are the additives: they are similar to the additives in a general game model except that the choice of which game to play occurs \emph{simultaneously} with the other bids. Thus for example in the game $\mathcal G^X_Y \with \mathcal H^U_V$ Abelard chooses a game and a bid for that game, but since Eloise bids simultaneously she must choose a bid for both games. Thus a winning strategy for Eloise in $\mathcal G^X_Y \with \mathcal H^U_V$ consists of a pair of bids $(x, u)$ together with winning strategies for both $\mathcal G^x_y$ and $\mathcal H^u_v$.

The denotation of the tensor product $\mathcal G^X_Y \otimes \mathcal H^U_V$ is more complicated. Eloise simply bids a pair $(x, u)$. Simultaneously Abelard must bid a pair of functions $f : U \to Y$ and $g : X \to V$, and then the games $\mathcal G^x_{fu}$ and $\mathcal H^u_{gx}$ are played in parallel in the sense specified by $R$. Similarly the exponential $\oc \mathcal G^X_Y$ is played as follows. Firstly Eloise chooses a bid $x \in X$. Then Abelard observes this and chooses a finite multiset $y_1, \ldots, y_n \in Y$. For each $y_i$ there is an exponential $\oc \mathcal G^x_{y_i}$, which will be similar to the parallel composition of infinitely many copies of $\mathcal G^x_{y_i}$. Then each of the $\oc \mathcal G^x_{y_i}$ is played in parallel (but typically a different sense of parallel than is used for exponentials), leading to $n \cdot \omega$ games being played in parallel. However the notion of winning strategy for Eloise in these games will depend on exactly what notion of parallelism is used in $R$, so it is difficult to say more in general.

As explained above, in some cases it is possible to consider this as a game played by two pairs of partners with a message-passing protocol, but in general it is necessary to consider functions which can depend on other functions in a higher-order way. Thus from a game-semantic perspective it will be more satisfying to replace the category of sets with a different locally cartesian closed category in which functions contain only a finite amount of information. Particularly interesting would be to use the recent work in progress of Abramsky and Jagadeesan on game semantics of dependent type theory. This would lead to a two-layered game model in which the bidding round has finer structure and the bids themselves specify strategies for sub-games. The difficulty would be to find a suitable sense in which $R$ is enriched and fibered over the model of dependent types.

\section{Relative completeness for additive-free fragments}\label{complete}

\begin{definition}[Complete model]
Let $\mathcal R$ be a model of $\LL$. A mapping from atoms to objects of $R$ is called a valuation in $\mathcal R$. Given a valuation $v$, we can extend it inductively to an interpretation of formulas in $\mathcal R$, denoted $\llbracket \varphi \rrbracket_v$ or simply $\llbracket \varphi \rrbracket$.

$\mathcal R$ is called a complete model of $\LL$ if for all formulas $\varphi, \psi$, if $\hom_R (\llbracket \varphi \rrbracket_v, \llbracket \psi \rrbracket_v)$ is nonempty for all valuations $v$ then the sequent $\varphi \vdash \psi$ is derivable in $\LL$. Completeness for $\MELL$ and other fragments is defined similarly.
\end{definition} 

A \emph{characterisation theorem} for a functional interpretation is a result saying that the equivalence between $\varphi$ and its functional interpretation $\exists x \forall y . \dial{\varphi}{x}{y}$ is derivable in some system, usually a base language like $\mathbf{HA}^\omega$ extended with \emph{characterisation principles}, which are axioms validated by the functional interpretation such as the axiom of choice, Markov's principle and independence of premise. In order to obtain the statement of the following lemma we take the logical formula
\[ \varphi \leftrightarrow \exists x \forall y . \dial{\varphi}{x}{y} \]
and split the bi-implication into its defining conjunction, then in each part we prenex the quantifiers and interpret them as dependent types.

The characterisation theorem for classical linear logic in \cite{oliva08} uses not $\exists x \forall y$ but a Henkin quantifier, $\aeq^x_y$, and so this `rearrangement' is unsound. The result we see is that this lemma fails to extend from $\MELL$ to $\LL$. (Given that this simultaneity is at the heart of the functional interpretations of classical linear logic, it is remarkable that this method works at all.) See section \ref{additives} for a discussion of how to extend the completeness theorem to include additives by correctly interpreting the simultaneous quantifier.

\begin{lemma}
Let $\mathcal R$ be a model of $\MELL$ and let $v$ be a valuation in $\mathcal R$. Let $\varphi$ be a formula of $\MELL$ with interpretation $\dial{\varphi}{X}{Y}$ in $\fd (\mathcal R)$, where the interpretation of an atomic proposition is
\[ \dial{p}{\{ * \}}{\{ * \}} : \binom{*}{*} \mapsto v (p) \]
Then the types
\[ \sum_{x : X} \prod_{y : Y} \hom_R (\llbracket \varphi \rrbracket, \dial{\varphi}{x}{y}) \]
and
\[ \sum_{y : Y} \prod_{x : X} \hom_R (\dial{\varphi}{x}{y}, \llbracket \varphi \rrbracket) \]
are inhabited.
\end{lemma}

\begin{proof}
These are proved simultaneously by induction on $\varphi$. In the base case we have $\varphi = p$ is an atom, and the point $*$ and identity morphism witnesses both (1) and (2).

In the negation case for (1) the inductive hypothesis for (2) gives $y \in Y$ together with morphisms $\pi_x : \hom_R (\dial{\varphi}{x}{y}, \llbracket \varphi \rrbracket)$. Then $\pi_x^\bot : \hom_R (\llbracket \varphi^\bot \rrbracket, \dial{\varphi^\bot}{y}{x})$. The case for (2) is symmetric.

For (1) of $\otimes$ the inductive hypothesis gives $x$ and $u$ together with morphisms $\pi_y : \hom_R (\llbracket \varphi \rrbracket, \dial{\varphi}{x}{y})$ and $\sigma_v : \hom_R (\llbracket \psi \rrbracket, \dial{\psi}{u}{v})$. Then for each $f : U \to Y$ and $g : X \to V$ we have
\[ \pi_{fu} \otimes \sigma_{gx} : \hom_R \left( \llbracket \varphi \otimes \psi \rrbracket, \dial{\varphi \otimes \psi}{x,u}{f,g} \right) \]

For (2) of $\otimes$ the inductive hypothesis gives $y$ and $v$ together with morphisms $\pi_x : \hom_R (\dial{\varphi}{x}{y}, \llbracket \varphi \rrbracket)$ and $\sigma_u : \hom_R (\dial{\psi}{u}{v}, \llbracket \psi \rrbracket)$. Define $f : U \to Y$ by $fu = y$ and $g : X \to V$ by $gx = v$. Then for each $(x,u)$ we have
\[ \pi_x \otimes \sigma_u : \hom_R \left( \dial{\varphi \otimes \psi}{x,u}{f,g}, \llbracket \varphi \otimes \psi \rrbracket \right) \]

For (1) of $\oc$, by the inductive hypothesis we have $x$ together with morphisms in $\pi_y : \hom_R (\sembrack{\varphi}, \dial{\varphi}{x}{y})$. Let $f : X \to Y^*$. We have
\[ \bigotimes_{y \in fx} \oc \pi_y : \hom_R \left( \bigotimes_{y \in fx} \oc \llbracket \varphi \rrbracket, \bigotimes_{y \in fx} \oc \dial{\varphi}{x}{y} \right) \]
Since $\mathcal R$ is a model of $\MELL$ we have $\bigotimes_{y \in fx} \oc \llbracket \varphi \rrbracket \cong \oc \llbracket \varphi \rrbracket$ and we are done.

For (2) of $\oc$, by the inductive hypothesis we have $y$ together with morphisms $\pi_x : \hom_R (\dial{\varphi}{x}{y}, \llbracket \varphi \rrbracket)$. Take $f$ to be the constant function returning the singleton multiset containing $y$. Then we have
\[ \oc \pi_x : \hom_R (\bigotimes_{y \in fx} \oc \dial{\varphi}{x}{y}, \oc \llbracket \varphi \rrbracket) \]
and we are done.
\end{proof}

\begin{theorem}[Relative completeness]
Let $\mathcal R$ be a model of $\MELL$ and let $\varphi$ be a formula of $\MELL$ which is true in $\fd (\mathcal R)$. Then $\varphi$ is true in $\mathcal R$.
\end{theorem}

\begin{proof}
Let $v$ be a valuation in $\mathcal R$, and let $\varphi$ be a formula of $\MELL$ with interpretation $\dial{\varphi}{X}{Y}$ in $\fd (\mathcal R)$ using the same interpretation of atomic propositions defined in the lemma. Since $\varphi$ is true in $\fd (\mathcal R)$ we have a winning bid $x : X$ together with winning strategies
\[ \pi_y : \hom_R (1, \dial{\varphi}{x}{y}) \]
From (2) of the lemma we have $y : Y$ together with winning strategies
\[ \sigma_x : \hom_R (\dial{\varphi}{x}{y}, \sembrack{\varphi}) \]
Therefore
\[ \sigma_x \circ \pi_y : \hom_R (1, \sembrack{\varphi}) \]
Since this holds for every valuation, $\varphi$ is true in $\mathcal R$.
\end{proof}

Let $\HO$ be the category of Hyland-Ong games and history-free, uniformly winning strategies \cite{hyland93}, with the the identity functor considered as an exponential. Then $\fd (\HO)$ is the model of `Hyland-Ong games with bidding'. (As a linear-nonlinear adjunction, the model of Hyland-Ong games has $R = S = \HO$, and $L = M$ is the identity functor.)

\begin{corollary}
$\fd (\HO)$ is a sound model of $\LL$ and a complete model of $\MLL$.
\end{corollary}

Notice that because the posetal reflection of $\HO$ is a \emph{lineale} in the sense of \cite{depaiva91b} (including having a trivial exponential), the category $\fd (\HO)$ is an example of the construction in that paper (modulo size issues). However examples of this kind have not been considered before, and in particular the completeness result is new.

Let $\mathbf{AG}$ be the category $Z$ of asynchronous games and (equivalence classes of) innocent winning strategies \cite{mellies04}. This is a sound model of $\LL$ which is proven in \cite{mellies05} to be complete for $\MELL$. That paper also provides a small variation which is complete for $\LL$, although using that model will not be necessary for our purposes.

\begin{corollary}
$\fd (\mathbf{AG})$, the category of asynchronous games with bidding, is a sound model of $\LL$ and a complete model of $\MELL$.
\end{corollary}

A large part of the motivation for this paper is to introduce the category $\fd (\mathbf{AG})$ and prove its soundness. It is an interesting model which will be studied in detail by the author in a follow-up paper: in particular there is a way to analyse the formulas containing additives which are valid in the model. See section \ref{additives} for a summary of the argument.

\section{$\fd$ is a functor}\label{functor}

Given a model $\R$ of $\MELL$, presented as a linear-nonlinear adjunction, we have defined a model $\fd (\R)$ of $\LL$. Since a collection of models forms a category we can ask whether $\fd$ is a functor. The answer is `yes' for the strongest notion of a morphism of models: a pair of functors which commute with all of our structure. Results of this kind are standard, and appear as early as \cite{scott78}. In the next section we will need a weaker notion of morphism of models of $\MLL$, namely lax monoidal functors between $*$-autonomous categories.

Since models are pairs of structured categories, they moreover form a 2-category, with 1-cells given by pairs of monoidal functors satisfying suitable conditions, and 2-cells given by pairs of natural transformations. We will leave the consideration of 2-categorical issues for later work, but it should be noted that most of the diagrams in this section and the next commute only up to natural isomorphism.

This section and the next do not contain all cases of the proofs (which would take another paper), but highlight the most interesting cases. Most of the proofs amount to showing that certain (sometimes quite formidable) dependent types are inhabited, and thus are natural candidates for formalisation in a dependently typed programming language, with suitable libraries for monoidal category theory and 2-category theory. The author intends to carry this out in the future.

\begin{definition}[Morphism of linear-nonlinear adjunctions]
Let $L \dashv M : R \to S$ and $L' \dashv M' : R' \to S'$ be linear-nonlinear adjunctions. A morphism $(F,G)$ from the former to the latter consists of functors
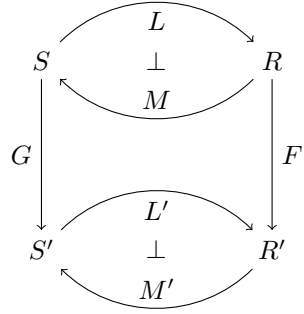
\begin{center}
\begin{tikzpicture}
\node (A) {$S$};
\node (B) [right=1cm of A] {$\bot$};
\node (C) [right=1cm of B] {$R$};
\node (D) [below=2cm of A] {$S'$};
\node (E) [below=2cm of B] {$\bot$};
\node (F) [below=2cm of C] {$R'$};
\draw [->] (A) edge [bend left=45] node [below] {$L$} (C);
\draw [->] (C) edge [bend left=45] node [above] {$M$} (A);
\draw [->] (D) edge [bend left=45] node [below] {$L'$} (F);
\draw [->] (F) edge [bend left=45] node [above] {$M'$} (D);
\draw [->] (A) to node [left] {$G$} (D);
\draw [->] (C) to node [right] {$F$} (F);
\end{tikzpicture}
\end{center}
such that
\begin{enumerate}
\item $F$ is a monoidal functor
\item $F$ and $G$ are cartesian monoidal functors
\item The following diagram commutes:
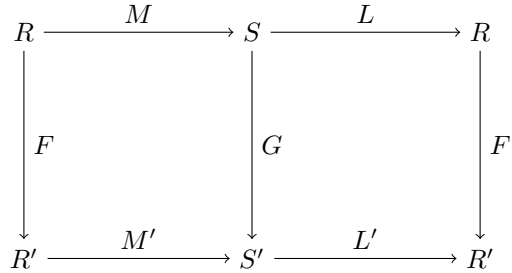
\begin{center}
\begin{tikzpicture}[node distance=3cm, auto,
implies/.style={double distance=4pt,-implies}
]
\node (A) {$R$};
\node (B) [right of=A] {$S$};
\node (C) [right of=B] {$R$};
\node (D) [below of=A] {$R'$};
\node (E) [below of=B] {$S'$};
\node (F) [below of=C] {$R'$};
\draw [->] (A) to node {$M$} (B);
\draw [->] (B) to node {$L$} (C);
\draw [->] (A) to node {$F$} (D);
\draw [->] (B) to node {$G$} (E);
\draw [->] (C) to node {$F$} (F);
\draw [->] (D) to node {$M'$} (E);
\draw [->] (E) to node {$L'$} (F);
\end{tikzpicture}
\end{center}
\end{enumerate}
(If we weaken this to having natural transformations $M' \circ F \implies G \circ M$ and $L' \circ G \implies F \circ L$ we obtain the linear-nonlinear equivalent of the `map of models' of \cite{hyland01}.)

The category of linear-nonlinear adjunctions and morphisms will be called $\LL\Mod$. The (larger) category of linear-nonlinear adjunctions in which $R$ and $R'$ do not necessarily have products (and $F$ is not necessarily cartesian monoidal) will be called $\MELL\Mod$. There is a forgetful functor $U : \LL\Mod \to \MELL\Mod$.
\end{definition}

\begin{lemma}
$\fd$ is a functor $\MELL\Mod \to \LL\Mod$.
\end{lemma}

\begin{proof}
We need to prove that
\begin{center}
\begin{tikzpicture}
\node (A) {$\fd_i (S)$};
\node (B) [right=1cm of A] {$\bot$};
\node (C) [right=1cm of B] {$\fd_l (R)$};
\node (D) [below=2cm of A] {$\fd_i (S')$};
\node (E) [below=2cm of B] {$\bot$};
\node (F) [below=2cm of C] {$\fd_l (R')$};
\draw [->] (A) edge [bend left=45] node [below] {$\fd_{dn} (L)$} (C);
\draw [->] (C) edge [bend left=45] node [above] {$\fd_f (M)$} (A);
\draw [->] (D) edge [bend left=45] node [below] {$\fd_{dn} (L')$} (F);
\draw [->] (F) edge [bend left=45] node [above] {$\fd_l (M')$} (D);
\draw [->] (A) to node [left] {$\fd_f (G)$} (D);
\draw [->] (C) to node [right] {$\fd_f (F)$} (F);
\end{tikzpicture}
\end{center}
is a morphism of $\LL\Mod$, given that $(F, G)$ is a morphism of $\MELL\Mod$.

We will prove the conditions for exponentials, namely that we have commuting squares
\begin{center}
\begin{tikzpicture}[node distance=3cm, auto,
implies/.style={double distance=4pt,-implies}
]
\node (A) {$\fd_l (R)$};
\node (B) [right of=A] {$\fd_i (S)$};
\node (C) [right of=B] {$\fd_l (R)$};
\node (D) [below of=A] {$\fd_l (R')$};
\node (E) [below of=B] {$\fd_i (S')$};
\node (F) [below of=C] {$\fd_l (R')$};
\draw [->] (A) to node {$\fd_f (M)$} (B);
\draw [->] (B) to node {$\fd_{dn} (L)$} (C);
\draw [->] (A) to node {$\fd_f (F)$} (D);
\draw [->] (B) to node {$\fd_ f (G)$} (E);
\draw [->] (C) to node {$\fd_f (F)$} (F);
\draw [->] (D) to node {$\fd_f (M')$} (E);
\draw [->] (E) to node {$\fd_{dn} (L')$} (F);
\end{tikzpicture}
\end{center}
For the left hand square let $\cg^X_Y \in \fd_l (R)$. We have
\[ ((\fd_f (M') \circ \fd_f (F)) (\cg))^X_Y : \binom{x}{y} \mapsto (M' \circ F) (\cg^x_y) \]
\[ ((\fd_f (G) \circ \fd_f (M)) (\cg))^X_Y : \binom{x}{y} \mapsto (G \circ M) (\cg^x_y) \]
These are equivalent using the identity functions on $X$ and $Y$ and the natural isomorphism $M' \circ F \cong G \circ M$. For the right hand square let $\cg^X_Y \in \fd_i (S)$. Then we have
\[ ((\fd_{dn} (L') \circ \fd_f (G)) (\cg))^X_{X \to Y^*} : \binom{x}{f} \mapsto \bigotimes_{y \in fx} (L' \circ G) (\cg^x_y) \]
\[ ((\fd_f (F) \circ \fd_{dn} (L)) (\cg))^X_{X \to Y^*} : \binom{x}{f} \mapsto F \left( \bigotimes_{y \in fx} L (\cg^x_y) \right) \]
Using the identity functions on $X$ and $X \to Y^*$ together with the natural isomorphism $L' \circ G \cong F \circ L$ and the fact that $F$ is monoidal we have natural transformations
\[ \bigotimes_{y \in fx} (L' \circ G) (\cg^x_y) \cong \bigotimes_{y \in fx} (F \circ L) (\cg^x_y) \cong F \left( \bigotimes_{y \in fx} L (\cg^x_y) \right) \]
\end{proof}

\section{$\fd$ is not a monad}\label{monad}

We have defined $\fd$ as a functor $\MELL\Mod \to \LL\Mod$. By composing with the forgetful functor in the opposite direction we obtain an endofunctor on $\MELL\Mod$. In this section we will investigate a monad-like structure on $\fd$. The starting point is the observation that there is a family of functors $\mu_R : \fd_l^2 (R) \to \fd_l (R)$ which appears to be the multiplication of a monad. In this section we investigate this structure and show that, on the contrary, $\fd$ is not a monad. The functors $\mu_R$ behave badly with respect to exponentials, and the corresponding functors $\mu_S : \fd_i^2 (S) \to \fd_i (S)$ cannot be defined in a reasonable way. Even when restricting to just $\MLL$, the functors $\mu_R$ are only lax monoidal, and the second monad law fails to hold, even in a lax way.

The main theorem of \cite{hofstra11}, which gives a sense in which the dialectica interpretation is a pseudo-monad, is extremely closely related. There are two main differences, other than the fact that our dialectica categories are far less general. The first is that Hofstra's multiplication operator, from a game-semantic point of view, treats the two players asymmetrically, and so appears to be incompatible with classical linear logic. The second is that, by using linear-nonlinear semantics, we insist on soundness for linear logic with exponentials. Nevertheless the second monad law does not appear to rely on either of these facts, which implies that the constructions are more different than they appear.

This section is interesting for two reasons. Firstly by replacing the term `functor' with `proof translation' the fact that functional interpretations fail to be monads becomes a fact about proof theory, essentially that functional interpretations do not commute as much as possible with other proof translations. Secondly the `multiplication' operator $\mu_R$ is actually important in the study of the dialectica interpretation of additives, as explained in the next section. Fortunately, although some of the types in this section are formidable, the action of the $\mu$ operator on objects is simple and intuitive.

There are several parts to the construction (again, without considering 2-categorical aspects). Firstly we describe the unit $\eta_\R : \R \to \fd (\R)$, which is a map of models. Then we describe $\fd^2 (\R)$ explicitly, and explore the multiplication operation. The resulting setup is illustrated in figure \ref{monaddiag}. Finally we must consider the monad laws. In practice we will focus on the parts which are both interesting (in particular, the cases which fail), and are practical to write by hand.

\begin{figure}
\begin{center}
\begin{tikzpicture}
\node (A) {$S$};
\node (B) [right=1cm of A] {$\bot$};
\node (C) [right=1cm of B] {$R$};
\node (D) [below=2cm of A] {$\fd_i (S)$};
\node (E) [below=2cm of B] {$\bot$};
\node (F) [below=2cm of C] {$\fd_l (R)$};
\node (G) [below=2cm of D] {$\fd_i^2 (S)$};
\node (H) [right=6.5mm of G] {$\bot$};
\node (I) [below=2cm of F] {$\fd_l^2 (R)$};
\draw [->] (A) edge [bend left=45] node [below] {$L$} (C);
\draw [->] (C) edge [bend left=45] node [above] {$M$} (A);
\draw [->] (A) to node [left] {$\eta_S$} (D);
\draw [->] (C) to node [right] {$\eta_R$} (F);
\draw [->] (D) edge [bend left=45] node [below] {$\fd_{dn} (L)$} (F);
\draw [->] (F) edge [bend left=45] node [above] {$\fd_f (M)$} (D);
\draw [->] (I) to node [right] {$\mu_R$} (F);
\draw [->] (G) edge [bend left=45] node [below] {$\fd_{dn}^2 (L)$} (I);
\draw [->] (I) edge [bend left=45] node [above] {$\fd_f^2 (M)$} (G);
\end{tikzpicture}
\end{center}
\caption{Unit and multiplication of $\fd$}
\label{monaddiag}
\end{figure}

We will begin with the natural transformation $\eta : I \to \fd$, where $I$ is the identity functor on $\MELL\Mod$. The functor
\[ \eta_R : R \to \fd_l (R) \]
takes an object $x \in R$ to the game with one play and outcome $x$,
\[ (\eta_R (x))^{\{ * \}}_{\{ * \}} : \binom{*}{*} \mapsto x \]
(recall that this is precisely the valuation of atoms in section \ref{complete}). It takes a morphism $\pi : \hom_R (x, y)$ to the strategy $(\id, \id, \pi)$ where $\id$ is the identity function on $\{ * \}$. The functor
\[ \eta_S : S \to \fd_i (S) \]
is similar. To be clear about notation, the components of $\eta$ are $\eta_\R$, where $\eta_\R$ is a lax morphism of models consisting of the functors $(\eta_R, \eta_S)$.

\begin{lemma}
$\eta$ is a well-defined natural transformation $\mathbf I \to \fd$.
\end{lemma}

Next we explicitly find $\fd^2 (\mathcal R)$ as a model of $\MELL$. An object $\mathcal G^X_Y$ of $\fd_l^2 (R)$ consists of sets $X$ and $Y$ together with a family of objects $\mathcal G^x_y$ of $\fd_l (R)$. Each such $\mathcal G^x_y$ itself has the form $(\mathcal G^x_y)^{U^x_y}_{V^x_y}$, where $U^x_y$ and $V^x_y$ are families of sets dependent on $x$ and $y$, and we have a family of objects $(\mathcal G^x_y)^u_v$ of $R$. This defines the objects of both categories $\fd_l^2 (R)$ and $\fd_i^2 (S)$.

Consider objects $\mathcal G^X_Y$ and $\mathcal H^W_Z$ of $\fd_l^2 (R)$ given by $(\mathcal G^x_y)^{U^x_y}_{V^x_y}$ and $(\mathcal H^w_z)^{P^w_z}_{Q^w_z}$, and consider a morphism from $\mathcal G$ to $\mathcal H$. This consists of functions $f : X \to W$ and $g : Z \to Y$ together with morphisms from $\mathcal G^x_{gz}$ to $\mathcal H^{fx}_z$ in $\fd_l (R)$. Each such morphism itself consists of functions $\alpha : U^x_{gz} \to P^{fx}_z$ and $\beta : Q^{fx}_z \to V^x_{gz}$ together with morphisms in $R$. Thus we have
\[ \hom_{\fd_l^2 (R)} (\mathcal G, \mathcal H) = \sum_{\substack{f : X \to W \\ g : Z \to Y}} \prod_{\substack{x : X \\ z : Z}} \sum_{\substack{\alpha : U^x_{gz} \to P^{fx}_z \\ \beta : Q^{fx}_z \to V^x_{gz}}} \prod_{\substack{u : U^x_{gz} \\ q : Q^{fx}_z}} \hom_R \left( (\mathcal G^x_{gz})^u_{\beta q}, (\mathcal H^{fx}_z)^{\alpha u}_q \right) \]

Morphisms in $\fd_i^2 (S)$ are much more complicated and will not be considered here. In order to complete the picture we would also need to consider $\fd_f^2 (M)$ and $\fd_{dn}^2 (L)$, but we will not do so here.

By thinking of $\fd^2 (R)$ as a game model the definition of $\mu_R$ becomes obvious. We begin with a game model $R$ of $\MLL$, and prepend a bidding round to obtain $\fd_l (R)$, then prepend an earlier bidding round to obtain $\fd_l^2 (R)$. A strategy for a game in this model consists of a bid in the first bidding round, together with a bid in the second bidding round for each possible bid of the opponent, and finally a strategy for each resulting game. This can be converted into a game with a single bidding round by bidding dependent types. Formally, given $\mathcal G^X_Y$ in $\fd_l^2 (R)$ given by $(\mathcal G^x_y)^{U^x_v}_{V^x_y}$, we define the object $\mu_R (\mathcal G)$ of $\fd_l (R)$ by
\[ \left( \mu_R (\mathcal G) \right)^{\sum_{x : X} \prod_{y : Y} U^x_y}_{\sum_{y : Y} \prod_{x : X} V^x_y} : \binom{x,f}{y,g} \mapsto (\mathcal G^x_y)^{fy}_{gx} \]

We will begin by showing that $\mu_R$ is lax monoidal but not strong monoidal. Suppose we have games $\cg^X_Y, \ch^W_Z \in \fd_l^2 (R)$ given by
\[ (\cg^x_y)^{U^x_y}_{V^x_y} : \binom{u}{v} \mapsto (\cg^x_y)^u_v \]
and
\[ (\ch^w_z)^{P^w_z}_{Q^w_z} : \binom{p}{q} \mapsto (\ch^w_z)^p_q \]
We need to construct a relative winning strategy
\[ \mu_R \cg \otimes \mu_R \ch \multimap \mu_R (\cg \otimes \ch) \]
We have
\[ (\mu_R \cg \otimes \mu_R \ch)^{\sum_{x : X} \sum_{x : X} \prod_{y : Y} U^x_y \times \sum_{w : W} \prod_{z : Z} P^w_z}_{\left( \sum_{w : W} \prod_{z : Z} P^w_z \to \sum_{y : Y} \prod_{x : X} V^x_y \right) \times \left( \sum_{x : X} \prod_{y : Y} U^x_y \to \sum_{z : Z} \prod_{w : W} Q^w_z \right)} \]
and
\[ (\mu_R (\cg \otimes \ch))^{\sum_{(x,w) : X \times W} \prod_{(f,g) : (W \to Y) \times (X \to Z)} (U^x_{fw} \times P^u_{gx})}_{\sum_{(f,g) : (W \to Y) \times (X \to Z)} \prod_{(x,w) : X \times W} ((P^u_{gx} \to V^x_{fw}) \times (U^x_{fw} \to Q^u_{gx}))} \]
To define a function
\[ \Phi : \sum_{x : X} \prod_{y : Y} U^x_y \times \sum_{w : W} \prod_{z : Z} P^w_z \to \sum_{(x, w) : X \times W} \prod_{\substack{(f, g) :\\ (W \to Y) \times (X \to Z)}} (U^x_{fw} \times P^u_{gx}) \]
suppose we are given $((x, \alpha), (w, \beta))$ where $\alpha : (y : Y) \to U^x_y$ and $\beta : (z : Z) \to P^w_z$. We need to define
\[ F : ((f, g) : (W \to Y) \times (X \to Z)) \to U^x_{fw} \times P^u_{gx} \]
which can be given by
\[ F (f, g) = (\alpha (fw), \beta (gx)) \]
In the other direction we need to define a function
\[ \Psi : \sum_{\substack{(f,g) :\\ (W \to Y) \times (X \to Z)}} \prod_{(x,w) : X \times W} ((P^u_{gx} \to V^x_{fw}) \times (U^x_{fw} \times Q^u_{gx})) \to \]
\[ \left( \sum_{w : W} \prod_{z : Z} P^w_z \to \sum_{y : Y} \prod_{x : X} V^x_y \right) \times \left( \sum_{x : X} \prod_{y : Y} U^x_y \to \sum_{z : Z} \prod_{w : W} Q^w_z \right) \]
Consider the left projection of this function (the right projection is symmetric). As input we are given the data
\begin{align*}
f &: W \to Y \\
g &: X \to Z \\
F &: ((x, w) : X \times W) \to ((P^w_{fx} \to V^x_{fw}) \times (U^x_{fw} \times Q^u_{gx})) \\
w &: W \\
h &: (z : Z) \to P^w_z
\end{align*}
We must produce $y : Y$ and $h' : (x : X) \to V^x_y$. We take $y = fw$ and
\[ h'x = \pi_L (F (x, w)) (h (gx)) \]
Note that neither of $\Phi$ and $\Psi$ can be canonically reversed, so $\mu_R$ is not strong monoidal.

Now, however, we consider the pair of games $\oc \mu_R (\cg)$ and $\mu_R (\oc \cg)$. The former is
\[ (\oc \mu_R (\cg))^{\sum_{x:X} \prod_{y:Y} U^x_y}_{\sum_{x:X} \prod_{y:Y} U^x_y \to \left( \sum_{y:Y} \prod_{x:X} V^x_y \right)^*} : \binom{x,f}{F} \mapsto \bigotimes_{(y,g) \in F(x,f)} \oc (\cg^x_y)^{fy}_{gx} \]
The latter, which takes some work to calculate, is
\[ (\mu_R (\oc \cg))^{\sum_{x:X} \prod_{f : X \to Y^*} \prod_{y \in fx} U^x_y}_{\sum_{f : X \to Y^*} \prod_{x:X} \left( \prod_{y \in fx} U^x_y \to \prod_{y \in fx} (V^x_y)^* \right)} : \binom{x,F}{f,G} \mapsto \bigotimes_{y \in fx} \bigotimes_{v \in Gx(Ff)y} \oc (\cg^x_y)^{Ffy}_v \]
The counter-witness types of these are incomparable, in the sense that there is no function in either direction which is natural in the types. Therefore we can say that $\mu : \fd_l^2 \to \fd_l$ is a well-defined natural transformation on the category of models of $\MLL$ and lax morphisms, but does not extend to $\MELL$.

The linear-nonlinear semantics gives us a better perspective on this problem. We can think of objects of $\fd_i (S)$ as games with bidding, but in which in the bidding round Abelard has the advantages granted by the exponential, namely he can observe Eloise's move and then choose several possible moves. In particular, the sequentiality of the bidding prevents us from extending our intuition about $\mu_R$ to $\fd_i^2 (S)$. A compound game in $\fd_i^2 (S)$ has two bidding rounds which are each played sequentially, and so bids are made in the order $\exists \forall \exists \forall$. We cannot reduce this to a single round of dependent bidding, because there is no way to specify that Abelard's first bid cannot depend on Eloise's second bid.

Restricting to $\MLL$, the first monad law holds up to natural isomorphism.

\begin{theorem}
There are natural isomorphisms
\center{\begin{tikzpicture}[node distance=3cm, auto]
\node (A) {$\fd_l (R)$};
\node (B) [right of=A] {$\fd_l^2 (R)$};
\node (C) [below of=A] {$\fd_l^2 (R)$};
\node (D) [below of=B] {$\fd_l (R)$};
\draw [->] (A) to node {$\eta_{\fd_l (R)}$} (B);
\draw [->] (B) to node {$\mu_R$} (D);
\draw [->] (A) to node {$\fd_f (\eta_R)$} (C);
\draw [->] (C) to node {$\mu_R$} (D);
\draw [double equal sign distance] (A) to node {} (D);
\end{tikzpicture}}
\end{theorem}

\begin{proof}
Consider an object $\cg^X_Y$ of $\fd_l (R)$. We can directly compute:
\[ ((\mu_R \circ \eta_{\fd_l (R)})(\cg))^{\sum_* \prod_* X}_{\sum_* \prod_* Y} : \binom{*,f}{*,g} \mapsto \cg^{f*}_{g*} \]
and
\[ ((\mu_R \circ \fd_f (\eta_R))(\cg))^{\sum_x \prod_y \{ * \}}_{\sum_y \prod_x \{* \}} : \binom{x, f}{y,g} \mapsto \cg^x_y \]
These are both naturally isomorphic to $\cg^X_Y$.
\end{proof}

The second monad law
\begin{center}
\begin{tikzpicture}[node distance=3cm, auto]
\node (A) {$\fd_l^3 (R)$};
\node (B) [right of=A] {$\fd_l^2 (R)$};
\node (C) [below of=A] {$\fd_l^2 (R)$};
\node (D) [below of=B] {$\fd_l (R)$};
\draw [->] (A) to node {$\mu_{\fd_l (R)}$} (B);
\draw [->] (B) to node {$\mu_R$} (D);
\draw [->] (A) to node {$\fd_f (\mu_R)$} (C);
\draw [->] (C) to node {$\mu_R$} (D);
\end{tikzpicture}
\end{center}
fails, even in a lax way (that is, this diagram does not contain a 2-cell). Consider an object of $\fd_l^3 (R)$ given by
\[ \cg : \binom{x : X}{y : Y} \mapsto \binom{u : U^x_y}{v : V^x_y} \mapsto \binom{p : (P^x_y)^u_v}{q : (Q^x_y)^u_v} \mapsto ((\cg^x_y)^u_v)^p_q \]
We can directly compute
\[ (\fd_f (\mu_R)) (\cg) : \binom{x : X}{y : Y} \mapsto \binom{u, \alpha : \sum_{u : U} \prod_{v : V} (P^x_y)^u_v}{v, \beta : \sum_{v : V} \prod_{u : U} (Q^x_y)^u_v} \mapsto ((\cg^x_y)^u_v)^{\alpha u}_{\beta v} \]
Therefore
\[ (\mu_R \circ \fd_f (\mu_R)) (\cg)^{\sum_{x : X} \prod_{y : Y} \sum_{u : U} \prod_{v : V} (P^x_y)^u_v}_{\sum_{y : Y} \prod_{x : X} \sum_{v : V} \prod_{u : U} (Q^x_y)^u_v} : \binom{x, f, \alpha}{y, g, \beta} \mapsto ((\cg^x_y)^{fy}_{gx})^{\alpha y (gx)}_{\beta x (fy)} \]
We also get
\[ \mu_{\fd_l (R)} (\cg) : \binom{x, f : \sum_{x : X} \prod_{y : Y} U^x_y}{y, g : \sum_{y : Y} \prod_{x : X} V^x_y} \mapsto \binom{p : (P^x_y)^{fy}_{gx}}{q : (Q^x_y)^{fy}_{gx}} \mapsto ((\cg^x_y)^{fy}_{gx})^p_q \]
Then $(\mu_R \circ \mu_{\fd_l (R)}) (\cg)$ involves nested dependent types:
\[ (\mu_R \circ \mu_{\fd (R)}) (\cg)^{\sum_{(x, f : \sum_{x : X} \prod_{y : Y} U^x_y)} \prod_{(y, g : \sum_{y : Y} \prod_{x : X} V^x_y)} (P^x_y)^{fy}_{gx}}_{\sum_{(y, g : \sum_{y : Y} \prod_{x : X} V^x_y)} \prod_{(x, f : \sum_{x : X} \prod_{y : Y} U^x_y)} (Q^x_y)^{fy}_{gx}} : \binom{x, f, F}{y, g, G} \mapsto ((\cg^x_y)^{fy}_{gx})^{F(y,g)}_{G(x,f)} \]
There is a natural transformation
\[ \sum_{x : X} \prod_{y : Y} \sum_{u : U} \prod_{v : V} (P^x_y)^u_v \to \sum_{(x, f : \sum_{x : X} \prod_{y : Y} U^x_y)} \prod_{(y, g : \sum_{y : Y} \prod_{x : X} V^x_y)} (P^x_y)^{fy}_{gx} \]
defined by
\[ (x, f, \alpha) \mapsto (x, f, \lambda (y, g) . \alpha y (gx)) \]
However there is none in the opposite direction. Similarly there is a natural transformation
\[ \sum_{y : Y} \prod_{x : X} \sum_{v : V} \prod_{u : U} (Q^x_y)^u_v \to \sum_{(y, g : \sum_{y : Y} \prod_{x : X} V^x_y)} \prod_{(x, f : \sum_{x : X} \prod_{y : Y} U^x_y)} (Q^x_y)^{fy}_{gx} \]
but none in the opposite direction. As a result, there is no morphism of $\fd_l (R)$ in either direction between the games $(\mu_R \circ \fd_f (\mu_R))(\cg)$ and $(\mu_R \circ \mu_{\fd_l (R)}) (\cg)$.

\section{Towards the additives}\label{additives}

In this section we briefly look at the question of how the completeness result in section \ref{complete} should be extended to full $\LL$. The intuition is that we are trying to simulate the behaviour of the simultaneous quantifier in \cite{oliva08}, in order to find a better analogue to the characterisation theorem $\varphi \multimapboth \aeq^x_y \dial{\varphi}{x}{y}$. This is ongoing work by the author, and this section only outlines the method.

We extend the language of $\mathbf{MELL}$ as follows. For a double-indexed family of formulas $\dial{\varphi}{X}{Y}$ we freely add a formula called $\binom{\oplus x : X}{\with y : Y} \dial{\varphi}{x}{y}$. These new formulas are called simultaneous additives (they could also be called `Henkin additives', because simultaneous quantifiers are a special case of Henkin quantifiers). The definition is fully recursive, so the individual formulas $\dial{\varphi}{x}{y}$ may themselves be simultaneous additives. 

There is a single introduction rule for simultaneous additives. Suppose we have double-indexed families of formulas $\dial{\varphi_i}{X_i}{Y_u}$ for $1 \leq i \leq m$ and $\dial{\psi_j}{U_j}{V_j}$ for $1 \leq j \leq n$. For all functions
\begin{align*}
f_j &: \prod_{i'} X_{i'} \times \prod_{j' \neq j} V_{j'} \to U_j \\
g_i &: \prod_{i' \neq i} X_{i'} \times \prod_{j'} V_{j'} \to Y_i
\end{align*}
for $1 \leq i \leq m$ and $1 \leq j \leq n$ we have a proof rule
\begin{prooftree}
\AxiomC{$\Gamma, \left( \dial{\varphi_i}{x_i}{g_i (\vec x_{-i}, \vec v)} \right)_{i=1}^m \vdash \Delta, \left( \dial{\psi_j}{f_j (\vec x, \vec v_{-j})}{v_j} \right)_{j=1}^n$ for all $\vec x \in \prod_i X_i, \vec v \in \prod_j V_j$}
\UnaryInfC{$\Gamma, \left( \binom{\oplus x_i : X_i}{\with y_i : Y_i} \dial{\varphi_i}{x_i}{y_i} \right)_{i=1}^m \vdash \Delta, \left( \binom{\oplus u_j : U_j}{\with v_j : V_j} \dial{\psi_j}{u_j}{v_j} \right)_{j=1}^n$}
\end{prooftree}
There is a hypothesis for all tuples $\vec x, \vec v$, hence this rule is generally infinitary. (The proof rule in \cite{oliva08} on which this is based uses free variables for $\vec{x}$ and $\vec{v}$ instead; it might be necessary to impose a restriction that the subproofs are `uniform' in the parameters in some way.) The extended language will be called $\fd\LL$.

We extend the valuation of formulas in a model $\mathfrak D (\mathcal R)$ to include simultaneous additives. If each $\dial{\varphi}{x}{y}$ is a formula in the language of $\mathfrak D \mathbf{LL}$ with interpretation $\dial{\dial{\varphi}{x}{y}}{U^x_y}{V^x_y}$ then the interpretation of $\binom{\oplus x : X}{\with y : Y} \dial{\varphi}{x}{y}$ is given precisely by the $\mu$ operator:
\[ \dial{\binom{\oplus x : X}{\with y : Y} \dial{\varphi}{x}{y}}{\sum_{x : X} \prod_{y : Y} U^x_y}{\sum_{y : Y} \prod_{x : X} V^x_y} : \binom{x,f}{y,g} \mapsto \dial{\dial{\varphi}{x}{y}}{fy}{gx} \]
It is an open question what should be the semantics of simultaneous additives in an arbitrary category. If it exists, it must have properties of both a limit and a colimit, since it includes products and coproducts as special cases.

\begin{theorem}
Let $R$ be any category, then $\fd_l (R)$ validates the simultaneous additive introduction rule.
\end{theorem}

If we try to prove the equivalence of $\varphi$ and $\binom{\oplus x : X}{\with y : Y} \dial{\varphi}{x}{y}$, where $\varphi$ is a formula of $\LL$, we find that we need some additional principles beyond $\fd\LL$, corresponding to the characterising principles of a functional interpretation. Two of these are
\[ \binom{\oplus x : X, u : U}{\with f : Y^U, g : V^X} \dial{\varphi \otimes \psi}{x,u}{f,g} \multimap \binom{\oplus x : X}{\with y : Y} \dial{\varphi}{x}{y} \otimes \binom{\oplus u : U}{\with v : V} \dial{\psi}{u}{v} \]
and
\[ \binom{\oplus z : X + U}{\with y : Y, v : V} \dial{\varphi \oplus \psi}{z}{y,v} \multimap \binom{\oplus x : X}{\with y : Y} \dial{\varphi}{x}{y} \oplus \binom{\oplus u : U}{\with v : V} \dial{\psi}{u}{v} \]
The first is a propositional analogue of the \emph{parallel choice} principle in \cite{oliva10}, which itself is a generalisation of the independence of premise principle. Write $\fd\LL^\#$ for $\fd\LL$ extended with these axioms and others for the exponential. Then, by directly simulating the characterisation theorem for a functional interpretation it should be possible to prove that if $\R$ is sound and complete for $\MELL$ then $\fd (\R)$ is sound and complete for $\fd\LL^\#$. In particular, $\fd (\mathbf{AG})$ should be a sound and complete model of $\fd\LL^\#$.

We continue by a purely syntactic argument. We prove that $\fd\LL$ has full cut elimination, and is a conservative extension of $\LL$ by identifying the usual additives with suitable simultaneous additives. Now if we take a formula $\varphi$ in the language of $\LL$ which is validated by $\fd (\mathbf{AG})$, we know that $\varphi$ is derivable in $\fd\LL^\#$, with a proof potentially involving both cuts and the characterising principles. In particular, since $\varphi$ does not contain simultaneous additives, any simultaneous additives introduced in the proof by a characterising principle must be removed by a cut. By analysing the ways in which cut elimination can fail in the presence of characterising principles, it should be possible to identity axioms in the language of $\LL$ which are sound and complete for $\fd (\mathbf{AG})$.

As an example, consider the formula $\bot \otimes \top$. This is not provable in $\LL$, because $\varphi \otimes \psi$ is provable in $\LL$ iff $\varphi$ and $\psi$ are both provable (by cut elimination) and $\bot$ is not provable. However every dialectica model has $\bot \otimes \top \cong \top$, and in particular $\bot \otimes \top$ is validated. Moreover for this simple example we can generalise from the category of sets to any cartesian closed category, which is the minimum structure needed to prove soundness. Therefore we can say that there is no dialectica category which is both sound and complete for classical linear logic. 

On the positive side, dialectica categories are more often considered as models of \emph{intuitionistic} linear logic, in which both witness and counter-witness sets must be nonempty. Since this example does not apply in that setting, there is still a possibility that we can construct complete dialectica models of intuitionistic linear logic by this method.

\bibliographystyle{plain}
\bibliography{refs}

\end{document}